\documentclass[12pt,a4paper,titlepage,reqno]{amsart}
\usepackage{amsfonts}
\usepackage{latexsym}
\usepackage{graphicx}
\usepackage{amsmath}
\newtheorem{thm}{Theorem}
\newtheorem{cor}{Corollary}
\newtheorem{lem}{Lemma}

\theoremstyle{definition}
\newtheorem{defn}{Definition}
\theoremstyle{remark}
\newtheorem{rem}{Remark}
\numberwithin{equation}{section}

\newcommand{\R}{\mathbb{R}}

\begin{document}

\thispagestyle {empty}

\title[]{ A note on free time evolution of the quantum wave function and Optimal Transportation }{08 06 22}
\author{\bf {Laura M. Morato}}

\dedicatory{ Universit\`a di Verona, Dipartimento di Informatica; E-mail:  laura.morato@univr.it  morato.lauramaria@gmail.com}

\begin{abstract}
It is shown that, in the absence of nodes and under regularity assumptions, a solution in a finite interval of time of the free Schroedinger equation solves a minimization problem which is a stochastic generalization of the classical optimal transportation problem with quadratic cost.
\end{abstract}
\maketitle

%%% -------------------------------
\pagestyle{myheadings}
\thispagestyle{myheadings}
%\markboth{}
\markboth{L. M. Morato} {Free evolution}

%%% ------------------------------

\section{Introduction}

Consider the free Schroedinger equation on $\R^d$

\begin{equation}\label{Schroedinger}
i\partial_t \psi + \frac1 2 \nabla^2 \psi = 0, \quad \psi(x,0)=\psi_o
\end{equation}

Putting $\rho := |\psi|^2$ and denoting by $S$ the principal argument of $\psi$ we can write

$$
\psi = \rho^{\frac 1 2}\exp ^{iS}
$$
\noindent so that, for all $(x,t)$ such that $\rho(x,t)$ is different from zero, \eqref{Schroedinger} is equivalent to

\begin{equation}\label{Mad1}
\partial_t \rho(x,t) + \nabla (\rho(x,t) \nabla S(x,t) )= 0
\end{equation}

\begin{equation}\label{Mad2}
\partial_t S(x,t) + \frac 1 2 (\nabla S(x, t))^2 -
\frac 1 2 \frac {\nabla^2 \sqrt{\rho(x,t)}}{\sqrt{\rho(x,t)}} = 0
\end{equation}

\noindent which are the (free) Madelung-fluid equations.

 Then $(\rho,\nabla S)$  acts as a "fluid-dynamical couple" and the first Madelung equation represents its continuity equation.

We address the question of whether $(\rho,\nabla S)$ is optimal in some sense among all time dependent fluid-dynamical couples $(\rho',v')$ which belong to a non trivial set such that $\rho'_o=\rho_o$ and $\rho'_1=\rho_1$, with $\rho_o :=|\psi_o|^2$ and $\rho_1 :=|\psi_1|^2$.

We do this by exploiting some features of Nelson's Stochastic Mechanics (see \cite{C2} for a review).

Introducing the drift-field

$$
b[\rho,\nabla S] := \nabla S +\frac 1 2 \nabla \log \rho
$$

\noindent we know, thanks to a general result due to Carlen \cite{C1}, that, if the quantum energy is finite at $t=0$, then there exists a Markov diffusion process $q^b$ with drift field $b=b[\rho,\nabla S]$ , diffusion matrix equal to the identity matrix  and time dependent probability density $\rho$. We call this the ''Nelson diffusion associated to $\psi\equiv \rho^{\frac 1 2}\exp ^{i S}$ ".

To be more precise, let $\Omega$ be the set of continous functions from $[0,1]$
to $\mathbb{R}^d$. Let $\mathcal F$ denote the associated Borel $\sigma$-algebra
and the filtration $(\mathcal F_t)_t$ be defined in the natural way. Let also $X_o$ be equal to $\omega(0)$.
Then we know that there exists a
probability measure $\mathbb{P}$ and a standard Brownian Motion $W$ on
$(\Omega,\mathcal F,(\mathcal F_t)_t)$,
such that  the configuration process $X$ satisfies the equality

\begin{equation}
X_t = X_0 + \int_0^t b[\rho,\nabla S] (X_s, s) ds +  W_t , \quad \mathbb P a.s.
\end{equation}
\noindent and $X_t$ has a propability density equal to $\rho(.,t)$ for all $t\in [0,1]$. Then  Nelson's diffusion is defined by identifying $q^{b[\rho,\nabla S]}$ with $X$.

In this paper we consider only the very regular case when $\rho_o$ is smooth and strictly positive and  the stochastic differential equation with coefficients $(b[\rho,\nabla S],I)$ has a strong solution. The more general case, when in particular $\psi$ has nodes, is currently the subject of further work.
\vskip 2mm
To formulate the optimization problem, we introduce a suitable set $\Xi_o(\rho_o,\rho_1)$ of smooth time dependent fluid-dynamical couples which connect $\rho_o$ to $\rho_1$ (see Definition 1).

Then, for a given choice of $\mathbb P$ and $W$ on $(\Omega,\mathcal F,\{\mathcal F_t\}_t)$, such that the law of $X_o$ has probability density $\rho_o$ and $W$ is independent of $X_o$  we define, for any $(\rho,v)$ in $\Xi_o$, the Nelson diffusion $q^b$, $b\equiv b[\rho,v]$,
as the solution of the S.D.E.

\begin{equation}\label{SDE1}
q^b_t = X_0 + \int_0^t b (q^b_s, s) ds +  W_t
\end{equation}

Introducing the equipartition $\{t_i\}_{i=0}^{n}$ of $[0,1]$ and putting, with

 $b\equiv b[\rho,v]$,

$$
\Delta q^b_i:=\int _{t_i}^{t_{i+1}} b(q^b_s,s)ds + (W_{t_{i+1}} - W_{t_i})
$$

\noindent we consider for every $\omega \in \Omega$ the classical action in discrete time

\begin{equation}\label{classical}
n \sum _{i=0}^{n-1}(\Delta q^b_i (\omega)) ^2
\end{equation}

 We take the average with respect to the initial configurations and  all possible Brownian paths, and  leave $n$ going to infinity. Exploitng  Nelson's renormalization formula \cite{Nelson*} to get rid of the divergent term, one gets

\begin{equation}
\lim _{n\to \infty} n\mathbb E \sum _{i=0}^{n-1}(\Delta q^b_i) ^2 - nd= \int _0^1 \int _{\R^d} (v^2- (\frac 1 2\nabla \ln \rho)^2)\rho dx dt
\end{equation}

Then we consider on $\Xi_o(\rho_o,\rho_1)$ the functional

\begin{equation}
A^Q(\rho,v):=
\int _{\R^d}\int_0^1 (v^2-(\frac 1 2 \nabla ln \rho)^2)\rho dt dx
\end{equation}
\vskip 5mm

The functional $A^Q$ and equivalent expressions of it, usually generalized by adding terms related to scalar
and vector potentials and possibly extended to the case when the configuration space
is a Riemannian manifold, were considered in the literature on S.M., mainly during the
'80s and '90s. In fact, within a stochastic control approach, their critical points were shown to be
related to the solutions of Schroedinger's equation \cite{GM}\cite{Nelson2}.
A fluid-mechanical reformulation of part of the results given in \cite{GM} was proposed by
Loffredo \cite{L} and was frequently adopted in the literature on S.M..
An approach based on stochastic differential games was proposed in \cite{Pavon} and, recently,  a relationship with Fisher information  was also suggested \cite{Yang}.

\vskip 5mm

The problem of establishing whether the critical points of $A^Q$, and of its equivalent expressions,
correspond to minimizers or not has remained unsolved. The main difficulty comes from the non convexity of
the functional.

It is worth mentioning that a variational method which allows deriving the Schroedinger equation in the framework of Nelson's Stochastic Mechanics starting from a convex functional,  was proposed by Yasue in \cite{Yasue}. Unfortunately  in his approach the "variations of a Nelson diffusion" are assumed to be smooth functions of the diffusion itself. This allows to exploit a nice integration by parts formula. But, as a consequence,  the "varied motions" are  Markov diffusions with a new, non constant, diffusion coefficient. Thus the "varied motions", at variance with what happens within the stochastic control approach adopted in \cite{GM} starting from $A^Q$, they are no longer Nelson diffusions and this characteristic makes difficult formulating a minimization problem.

\vskip 5mm
In this work we construct on a proper space a convex functional such that a suitable restriction of it is equivalent, in a proper sense, to $A^Q$ (see \eqref{Finfinito},\eqref{eguaglianza} and \eqref{eguaglianza2}). We prove that, if $\psi \equiv \rho^{\frac 1 2} \exp i S$ , with satisfies \eqref{Schroedinger} with $|\psi_o|^2= \rho_o$, $|\psi_1|^2= \rho_1$ and $(\rho,\nabla S)\in \Xi_o(\rho_o,\rho_1)$, then

 \begin{equation}
A^Q(\rho,\nabla S)\leq A^Q(\rho',v'),\quad \forall (\rho',v')\in \Xi_o(\rho_o, \rho_1)
\end{equation} .

\newpage
\section{A convex asymptotic functional}

 We consider the probability space $(\Omega,\mathcal F,(\mathcal F_t)_t)$, where $\Omega$ is the set of continous functions from $[0,1]$ to $\mathbb{R}^d$, $\mathcal F$ denotes the associated Borel $\sigma$-algebra
and the filtration $(\mathcal F_t)_t$ is defined in the natural way.

 Let also $X_o := \omega(0)$. Let $\mathbb P$  be a probability measure such that the law of $X_o$  has   the probability density $\rho_o$ and let $W$ be a standard Brownian Motion independent of $X_o$.

\begin{rem}
Let $\rho_o$ be smooth and strictly positive and $b: R^d\times [0,1] \to \R^d$ a smooth time dependent drift field with sublinear growth at infinity. Then the Stochastic Differential Equation with coefficients $(b,I)$ has a unique strong solution, so that there exists a  unique continous square integrable Markov process $q^b$ s.t.

\begin{equation}\label{SDE}
q_t^b = X_0 + \int_0^t b(q_s^b , s) ds + W_t, \quad \mathbb P a.s.
\end{equation}
The process admits a time continous probability density $\rho$ which is smooth and strictly positive. Defining the "current velocity field" $v$ by $v:= b-\frac 1 2 \nabla \log \rho$ , the Fokker-Planck equation, describing the time evolution of the probability density $\rho$, takes the form the continuity equation for the pair $(\rho,v)$,i.e.

\begin{equation}\label{continuity}
\partial _t\rho + \nabla (\rho v)=0
\end{equation}
\noindent (see \cite{Nelson1}).
\end{rem}
\vskip 2mm

\begin{defn}
We denote by $\Xi_o$ the set of pairs $(\rho,v)$ where $\rho$ is a strictly positive smooth time dependent probability density on $\R^d$, $v$ is a smooth time dependent velocity field on $\R^d$ such that

\begin{equation}
\partial _t\rho + \nabla (\rho v)=0\quad \quad \text{(continuity equation)}
\end{equation}

\noindent and

\begin{equation}\label{finite action}
\int _0^1 \int _{\R^d} (v^2+ (\frac 1 2\nabla \ln \rho)^2)\rho dx dt < \infty,\quad \quad \text{(finite action condition)}
\end{equation}

Moreover  $v$ and $ \nabla \log \rho$ are assumed to have a sublinear growth at infinity.

\end{defn}
The set $\Xi_o$ is a small subset of the "set of proper infinitesimal characteristics" introduced by Carlen in \cite{C3}
\vskip 4mm

We now introduce the following space of processes
\begin{equation}\label{Ldue}
L^2_{[0,1]}(\mathbb{P}) := \{\beta:\Omega \times [0, 1]\rightarrow\mathbb{R}^d\; s.t.
\int_\Omega \int_0^1\beta^2(\omega,s) ds \mathbb P(d\omega) <\infty\}
\end{equation}

Let define

 $$
 q_t^\beta := X_0 + \int_0^t \beta_s ds +W_t ,\quad   X_0(\omega)=\omega(0)
 $$
  and
$$
\Delta q_i^\beta := q_{t_{i+1}}^\beta - q_{t_i}^\beta
$$

 Considering the equipartition $\{t_i\}_{i=0}^{n}$ of $[0,1]$ and denoting by $\mathbb E$ the integration with respect to $\mathbb P$, we introduce the convex functional $F_n: L_{[0,1]}^2 (\mathbb{P})\rightarrow \mathbb{R}$

\begin{equation}\label{Fn}
F_n (\beta) :=  n\mathbb E \sum_{i=0}^{n-1} (\triangle q_i^\beta)^2
\end{equation}

Let now $\beta$ be a "Markovian drift", i.e. such that
$\beta_t = \beta_t^b $ where

\begin{equation}\label{beta}
\beta_t^b:= b (q^b_t,t),
\end{equation}

\noindent $q^b$ being the solution of the stochastic differential equation \eqref{SDE}.

 Introducing the notation

\begin{equation}\label{brhov}
b[\rho,v] := v +\frac 1 2 \nabla \log \rho ,
\end{equation}

\noindent if $b$ is equal to $b[\rho,v]$  with $(\rho,v)\in \Xi_o$ , then, by the finite action condition, $\beta^b := (\beta ^b_t)_{t\in [0,1]}$  belongs to $L^2_{[0,1]}(\mathbb P)$.

Moreover, exploiting  Nelson's renormalization formula (\cite{Nelson*}) we can take the limit for $n$ going to infinity, getting, with $b\equiv b[\rho,v]$,

\begin{equation}\label{rin1}
\lim_{n\rightarrow\infty} n \mathbb{E} \sum_{i=0}^{n-1} (\triangle q_i^b)^2 - nd =
\mathbb{E}\int_0^1 (b^2 (q^b_t , t) + \nabla b (q_t^b, t))dt   \quad\quad
\end{equation}

and, integrating by parts,

\begin{equation}\label{rin2}
\mathbb{E}\int_0^1 (b^2(q_t^b, t)+\nabla b(q_t^b, t))dt = \int_0^1 \int_{\mathbb{R}^d}
 (v^2 - (\frac 1 2 \nabla ln \rho)^2)\rho dx dt <\infty   \quad\quad
\end{equation}

Then, for any "Markovian element" $\beta^b$,
$b\equiv b[\rho,v]$, $(\rho,v)$ in $\Xi_o$, we have

\begin{equation}
\lim_{n\rightarrow\infty} (F_n(\beta^b)- nd)< \infty
\end{equation}

\noindent and, for any $\lambda\in[0,1]$ and $(\beta^{b_1},\beta^{b_2}) \in L^2_{[0,1]}(\mathbb{P})$, $b_1 :=b[\rho_1,v_1]$ and $b_2 :=b[\rho_2,v_2]$ with $(\rho_1,v_1)$ and $ (\rho_2,v_2)$ in $\Xi_o$,

\begin{multline}
\lim_{n\rightarrow\infty}(F_n(\lambda\beta^{b_1} + (1 - \lambda)\beta^{b_2})- nd)\leq\\
 \leq \lim_{n\rightarrow\infty} \{\lambda (F_n(\beta^{b_1}) - nd) + (1-\lambda)(F_n(\beta^{b_2})- nd)\} < \infty
\end{multline}

Since the convex combination of three elements is equal to the convex combination of proper two elements, one can see by induction that for any finite convex combination
$\sum_{i=1}^m \alpha_i\beta^{b_i}$, $(\beta ^{b_i})_{i=1}^{m}$  being Markovian elements in $L_{[0,1]}^2 (\mathbb{P})$, we have
$$
\underset {n\to \infty} {\lim}[F_n (\sum_{i=1}^m \alpha_i \beta^{b_i}) - nd] < \infty
$$
Denoting by $\Sigma$ the convex  set given by all finite convex combinations of Markovian elements
in $L^2_{[0,1]}(\mathbb{P})$, we can define the convex functional

$$\hat{F}_\infty : L^2_{[0,1]}(\mathbb{P})\to \bar \R $$

\begin{equation}\label{Finfinito}
\hat F_\infty (\beta) := \begin{cases}
                               \underset {n\to \infty}{lim}(F_n(\beta)- nd)\quad &\forall\beta \in \Sigma \\
                               +\infty \quad \text {otherwise}
                             \end{cases}
\end{equation}

The elements of $\Sigma$ are not markovian in general and they are somehow reminiscent of the quantum mixtures, but in fact describing quantum mixtures would require an enlarged probability space (see for example \cite{CPM}).

\newpage
\section{Critical points and minima }

\begin{defn}

  Let $(\Omega,\mathcal F, (\mathcal F_t))_t)$, $(\mathbb P,W)$ and $X_o$ be defined as in the beginning of Section 2. Assume also $b\equiv b[\rho,v]$, with $(\rho,v)$ in $\Xi_o$ and let  $q^b $ be defined by \eqref{SDE}.

We define the functionals

$$I:\{b = b[\rho,v] :(\rho,v)\in \Xi_o     \}\to \R$$

\begin{equation}\label{I}
b \mapsto \mathbb{E}\int_0^1 (b^2 (q^b_t , t) + \nabla b (q_t^b, t))dt
\end{equation}

\noindent and,

$$
A^Q :\{ (\rho,v):\text{ finite action condition holds} \} \to \R
$$

\begin{equation}\label{A}
(\rho,v) \mapsto\int_0^1 \int_{\mathbb{R}^d}
 (v^2 - (\frac 1 2 \nabla ln \rho)^2)\rho dx dt
\end{equation}.

\end{defn}

\vskip 2mm

By \eqref{rin2} we have, $\forall (\rho,v)\in \Xi_o$
\begin{equation}\label{eguaglianza}
A^Q(\rho,v) = I(b[\rho,v])
\end{equation}

\noindent and, defining $\beta^{b[\rho,v]}$ by \eqref{beta}and \eqref{brhov},

\begin{equation}\label{eguaglianza2}
I(b[\rho,v]) = \hat F_\infty (\beta^{b[\rho,v]})
\end{equation}

Let $\Xi_o(\rho_o,\rho_1)$ be constituted by all elements $(\rho,v)$ in $\Xi_o$ such that $\rho(.,0)$ is equal to $\rho_o$ and $\rho(.,1)$ is equal to $\rho_1$. We face the problem of looking for possible  minima of $A^Q$ on $\Xi_o(\rho_o,\rho_1)$.

\vskip 2mm

As a first step we revisit the fluid dynamical version of the variational principle given in \cite{GM}, that was  proposed in \cite{L}.

\vskip 2mm

  For any sufficiently regular function $f$ of a pair $(\rho,v)$ in $\Xi_o$ and for any sufficiently regular pair  $(\delta \rho :(\R\times [0,1]\to \R),\delta v: (\R^d\times [0,1]\to \R^d))$   we will use the short-hand notation
 \begin{equation}\label{Dtilde}
\tilde D f(\rho,v) (\delta\rho,\delta v) := \underset{\epsilon \to 0}{lim} \frac 1 \epsilon [f (\rho + \epsilon \delta \rho, v + \epsilon \delta v)-f (\rho , v ) ]
\end{equation}

\begin{defn}
We will say that $(\rho,v)\in \Xi_o(\rho_o,\rho_1)$ is a critical point of $A^Q$ if

$$
\tilde D A^Q(\rho,v) (\delta\rho,\delta v)=0
$$
\noindent for all $(\delta \rho,\delta v)$ satisfying the conditions
\vskip 2mm
 a) $\delta \rho \in C_K^\infty (\mathbb{R}^d\times [0,1]\to \mathbb{R})$, $\delta v \in C_K^\infty(\R^d\times [0,1]\to \R ^d)$ with $\delta \rho_o=0$ and $\delta \rho_1=0$

 b) $ \frac {d}{dy} \{\partial_t (\rho + y \delta\rho) +\nabla [(\rho+y\delta \rho)(v+y\delta v)]\}|_{y=0} = 0$.
\vskip 2mm
\end{defn}

\begin{lem}

Let $(\delta \rho,\delta v)$ satisfy conditions a) and b).
\vskip 2mm

 Then a sufficient condition in order that an element $(\rho,v)$  of   $\Xi_o(\rho_0,\rho_1)$ satisfies the equality
\begin{equation}\label{derivata}
\tilde D A^Q (\rho, v) (\delta\rho, \delta v) = 0 ,\quad \forall (\delta \rho, \delta v)\quad \text{s.t. a) and b) hold}
\end{equation}
\noindent is that, for all $(x,t)\in \R^d\times [0,1]$
\vspace {2mm}

i) There exists a smooth $S : \mathbb{R}^d \times [0,1] \to  \R$ such that
$$
v(x,t) = \nabla S(x,t)
$$

ii)  $(\rho,\nabla S) $satisfies Madelung's equations, so that

\begin{equation}\label{Mad1}
\partial_t \rho(x,t) + \nabla (\rho(x,t) \nabla S(x,t) )= 0
\end{equation}

\begin{equation}\label{Mad2}
\partial_t S(x,t) + \frac 1 2 (\nabla S(x, t))^2 -
\frac 1 2 \frac {\nabla^2 \sqrt{\rho(x,t)}}{\sqrt{\rho(x,t)}} = 0
\end{equation}

\end{lem}

\begin{proof}

\vspace {2mm}

\vspace{2mm}

Let $\lambda : \mathbb{R}^d \times [0,1] \to  \R$  be smooth. Define

$$
F(\rho,v,\lambda) := A^Q (\rho, v) + \int_0^1 \int_{\R^d} \lambda (\partial_t \rho + \nabla (\rho v)) dx dt
$$

In the given assumptions one has, for every $(\rho,v)$ in $\Xi_o$

$$
F(\rho,v,\lambda) := A^Q (\rho, v)
$$
\noindent and, for all $(\delta \rho,\delta v)$ which satisfy a) and b),
$$
\tilde D F(\rho, v, \lambda) (\delta \rho, \delta v, 0) = \tilde D A^Q (\rho, v)(\delta \rho, \delta v)
$$

Expliciting $\tilde D $ and integrating by parts, one gets

$$
\tilde D F(\rho, v, \lambda) (\delta\rho, 0, 0) =
 $$
\begin{equation}\label{E1}
= \int_0^1 \int_{\R^d} [(-\partial_t \lambda - \nabla \lambda) v + v^2 + \frac {\nabla^2 \sqrt\rho}{\sqrt \rho}]\delta \rho dx dt+[\lambda\delta \rho]_0^1 +[\lambda v \delta \rho]_{-\infty}^{+\infty}
\end{equation}

\noindent and
$$
\tilde D F(\rho, v, \lambda) (0, \delta v, 0) =
$$
\begin{equation}\label{E2}
= \int_0^1 \int_{\R^d} (2 v - \nabla \lambda) \rho \delta v dx dt +[\lambda \rho \delta v]_{-\infty}^{+\infty}
\end{equation}

\noindent where all boundary terms are equal to zero. Then putting $\lambda  = 2 S $ and $v=\nabla S$ , we get, by \eqref{Mad2},
$$
\tilde D A^Q (\rho,\nabla S)(\delta \rho,\delta v)=\tilde D F(\rho, \nabla S, 2 S) (\delta\rho, \delta v, 0) = 0
$$

\vskip 4mm

\end{proof}

\begin{cor}

If the Schroedinger equation

\begin{equation}\label{Schroedinger2}
i\partial_t \psi + \frac1 2 \nabla^2 \psi = 0, \quad |\psi_o|^2 =\rho_o, \quad |\psi_1|^2 =\rho_1
\end{equation}

\noindent is satisfied by  $\psi:= \rho^{\frac 1 2}\exp (iS)$ with $(\rho,\nabla S)\in \Xi_o$ then \eqref{derivata} holds.
\end{cor}

\vskip 4mm

As a second step we show that Lemma 1 and the convexity of $\hat F_\infty$ allows to solve a minimization problem for $A^Q$.

\begin{thm}

Let $(\rho_o,\rho_1)$ be probability densities on $\R^d$ with finite variance.

Assume that $\psi:= \rho^{\frac 1 2}\exp^{iS}$ satisfies  the free Schroedinger equation

\begin{equation}
i\partial_t \psi + \frac1 2 \nabla^2 \psi = 0\quad,\quad |\psi_0|^2 = \rho_0 , |\psi_1|^2 = \rho_1,
\end{equation}

\noindent and that $(\rho,\nabla S)$ belongs to $ \Xi_o(\rho_o,\rho_1)$.

Then

\begin{equation}\label{inequality}
A^Q(\rho,\nabla S) \leq A^Q(\rho',v')
\end{equation}

$$
\forall (\rho',v')\in \Xi_o(\rho_o,\rho_1)
$$

\end{thm}

\begin{proof}

Let $(\rho',v')$ belong to $ \Xi_o(\rho_o,\rho_1)$. Put $g:= \rho'-\rho$. Being $\rho'$ and $\rho$  normalized to $1$ one has that $\int_{\R^d} g dx$ is equal to $0$ and that $\rho +yg$ is also normalized to $1$ and strictly positive for all $y\in [-1,1]$.
We consider firstly the case when $g$ is of class $C_K^\infty$.

Introducing a time dependent vector field $X_y^g$, we consider the family $(\rho'_y,v'_y)_{y\in [-1,1]}$  defined by

\begin{equation}\label{family}
\begin{cases}
\rho'_y := \rho +yg\\
v'_y := \nabla S + X_y^g
\end{cases}
\end{equation}

Requiring $(\rho'_y,v'_y)$ to satisfy, for all $y\in [-1,1]$, the continuity equation

\begin{equation}
\partial _t\rho'_y + \nabla (\rho'_y v'_y)=0
\end{equation}

\noindent and putting

$$
u:= (\rho +y g)(X_y^g)
$$

\noindent one gets
$$
\sum _{i=1}^{d}\partial_{x_i} u_i= -y[\partial _t g+\sum _{i=1}^{d}\partial_{x_i} (g \partial_{x_i} S)]
$$

This equation admits solutions $u$ in $C_K^\infty (\R^d\times [0,1]\to \R^d)$ so that $X_y^g$ can be chosen in $C_K^\infty (\R^d\times [0,1]\to \R^d)$.  Then $(\rho'_y,v'_y)$ belongs to $ \Xi_o(\rho_o,\rho_1)$  for all $y\in [-1,1]$.

\noindent Moreover one can check that $X_y^g$ depends smoothly on $y$ and that $\frac {d}{dy} X_y^g|_{y=o}$ also belongs to $C_K^\infty (\R^d\times [0,1]\to \R^d)$. Defining

$$
\delta \rho := g
$$

\noindent and

$$
\delta v:= \frac {d} {dy} X_y^g|_{y=0}
$$
\noindent one can see that conditions a) and b) in Lemma 1 are satisfied. Then

\begin{equation}\label{base}
\lim_{y\to 0}\frac 1 y (A^Q(\rho'_y,v'_y)-A^Q(\rho,\nabla S))=
 \tilde D A^Q(\rho,\nabla S)(\delta\rho,\delta v)=0
\end{equation}
\vskip2mm

 We can exploit \eqref{eguaglianza} to get

 \begin{equation*}
 \lim _{y \to 0}\frac 1 y(I(b[\rho'_y,v'_y])-  I(b[\rho,\nabla S]))=0
 \end{equation*}

 \noindent which, by \eqref{eguaglianza2}, is equivalent to

$$ \underset{y \to 0}{\lim}\frac 1 y( \hat F_\infty (\beta^{b[\rho'_y,v'_y]}) - \hat F_\infty (\beta ^{b[\rho, \nabla S]})) =0 $$

Since by the chain of equalities \eqref{eguaglianza} and \eqref{eguaglianza2}, for all $y\in[0,1]$,

\begin{equation}
\hat F_\infty (\beta^{b[\rho'_y,v'_y])}=A^Q(\rho'_y,v'_y)
\end{equation}

\noindent with $(\rho'_y,v'_y)$ defined by \eqref{family}, the dependence of
  $\hat F_\infty (\beta^{b[\rho'_y,v'_y]})$ on $y\in [-1,1]$ is smooth by construction.

Finally, by the convexity of $\hat F_\infty$,

\begin{equation}\label{inequality y}
\hat F_\infty (\beta ^{b[\rho, \nabla S]}) \leq \hat F_\infty (\beta^{b[\rho'_y,v'_y]}),\quad \forall y \in [- 1,1]
\end{equation}

 Then, again by the chain of equalities \eqref{eguaglianza} and \eqref{eguaglianza2} and putting $y=1$ we find

$$
A^Q(\rho,\nabla S)\leq A^Q(\rho',v')
$$

\noindent for all $(\rho',v')$ in $ \Xi_o(\rho_o,\rho_1)$ such that $\rho' := \rho +g$, $g\in C_K^\infty$.

\vskip 5mm

For a generic element $(\rho',v')$  of $ \Xi_o(\rho_o,\rho_1)$ put again $g:= \rho'-\rho$. Then g is smooth and such that $\int_{\R^d} g dx =0$ but its support is non necessarily compact.

 Let $(D_j)_j $ be a sequence of open bounded subsets of $\R^d\times [0,1]$ such that $D_j\uparrow \R^d\times [0,1]$.

 Let also $( h_j )_j \subset C_K^\infty (\mathbb{R}\times [0,1]\rightarrow \mathbb{R})$ be
such that $0\leq h_j \leq 1$ and
\begin{equation}\label{h}
h_j (x,t) = \begin{cases}
                  1     \qquad     (x,t) \in D_j \\
                  0     \qquad     (x,t) \in \underset{m=j+2}{U^\infty} (D_m \setminus D_{m-1})
               \end{cases}
\end{equation}

Then, putting $g_j:= gh_j$, $\rho'_j:=\rho + g_j$ and $v'_j:= v+X_1^{g_j}$, where $X_y^{g_j}$, $y\in [0,1]$, is constructed as before, one has, for all $j\in \mathbb N$,

$$
A^Q(\rho,\nabla S)\leq A^Q(\rho'_j,v'_j)=
$$
$$
=\int _{D_j}(v'^2-(\frac 1 2 \nabla \log \rho')^2)\rho' dx dt)+\int _{D_j^c}(v_j'^2-(\frac 1 2 \nabla \log \rho'_j)^2)\rho'_j dx dt)
$$

Then
$$
A^Q(\rho,\nabla S)\leq \lim_{j\to \infty} A^Q(\rho'_j,v'_j)= A^Q (\rho',v'))
$$
\noindent for all $(\rho',v')$ in $ \Xi_o(\rho_o,\rho_1)$.

\end{proof}

\vskip 4mm
\section{Comparison with the classical Optimal Transportation Problem }

The classical Monge-Kantorovich Optimal Transportation Problem \cite{M}\cite{K} with quadratic cost, in the Monge
formulation, is

\begin{equation}\label{M}
\underset{T:\;T \sharp\rho_o =\rho_1}{inf} \int _{\mathbb{R}^d}
\mid Tx-x\mid^2\rho_0(x)d x
\end{equation}

\noindent where the infimum is taken on the set of maps $T:\mathbb{R}^d \rightarrow \mathbb{R}^d$ that "transport $\rho_0$ onto $\rho_1$".

 In \cite{B} Br\'enier
proved that the solution $T_0$ of problem \eqref{M} is the unique map $T_o: T_o \sharp\rho_o =\rho_1$, which is
a gradient of a convex function $\phi$.

 Moreover, by introducing the "time variable" $t\in [0,1]$, in \cite{B.B.} a computational fluid-mechanical solution of  the Monge problem was constructed.

The new problem can be formulated  as follows (see \cite{V} chap 8):

\begin{equation*}
\underset{(\rho, v)\in \mathbb V(\rho_0,\rho_1)}{inf} \int_{\mathbb{R}^d}\int_0^1 v^2\rho dt dx
\end{equation*}

\vskip 2mm
\noindent where $\mathbb V(\rho_0, \rho_1)$ is a very general set of fluid-dynamical couples $(\rho,v)$ which in particular satisfy the continuity equation in distributional sense.The "set of proper infinitesimal characteristics" connecting $\rho_o$ to $\rho_1$ , and in particular $\Xi_o(\rho_o,\rho_1)$, are subsets of $\mathbb V(\rho_0, \rho_1)$.

 Assuming that $\rho_0$ and $\rho_1$ have finite variance and
 denoting by $\tau_2 (\rho_0, \rho_1)$ the infimum in \eqref{M}, we have the
B\'enamou-Br\'enier formula

$$
\tau_2 (\rho_0, \rho_1) = \underset{(\rho, v)\in
	V_(\rho_0,\rho_1)}{inf} A(\rho, v)
$$

\noindent where

$$
A(\rho,v) :=\int_{\mathbb{R}^d}\int_0^1 v^2\rho dt dx
$$
\noindent which, when restricted to $\Xi_o(\rho_o,\rho_1)$, is the classical limit of $A^Q$. By a simple change of variables the functional $A$ becomes convex.

 Moreover if $(\rho,v)$ is the  solution to the B\'enamou-Br\'enier problem, it turns to be a gradient-flow and, in the smooth case, the classical limit of the free Madelung fluid equations holds, which are equivalent to

$$
\partial _t\rho + \nabla (\rho v)=0
$$

$$
\partial _t v + v \nabla v =0,\quad v_o(x) = \nabla \phi(x) -x
$$
\noindent where $\nabla \phi$ is the solution of the Monge problem.


\newpage
\begin {thebibliography} {}



\bibitem{B}Br\'enier Y.Polar factorization and monotone rearrangement of vector-valued functions,
Communications on pure and applied mathematics, volume 44 (4), 375-417, (1991)

\bibitem{B.B.}B\'enamou J.D. and Y. Br\'enier Y. A computational fluid mechanics solution to the Monge-Kantorovich mass transfer problem, Numerische Mathematik volume 84, 375-393,(2000)


\bibitem{C1}Carlen E. Conservative diffusions, Communications in Mathematical Physics, volume 94, 293-315, (1984)

\bibitem{C2}Carlen E. Stochastic Mechanics: a look back and a look ahead, in Diffusion, Quantum Theory, and Radically Elementary Mathematics, Faris W. ed,117-139, (Princeton, NJ: Princeton University Press) (2006)
\bibitem{C3}Carlen E. Progress and problems in Stochastic Mechanics. In: Gielerak,R., Karwoski,W.(eds ) Stochastic Processes-Mathematics and Physics. Lecture notes in Mathematics. Vol.1158, pp.25-51.Berlin, Heidelberg,New York: Springer (1985)
\bibitem{CPM}Cufaro Petroni N. and Morato L.M. Entangled states in Stochastic Mechanics. Physica A: Math. Gen.{\bf 33}, 5833-5848 (2000)


\bibitem{GM} Guerra F. and Morato L.: Quantization of
Dynamical Systems and Stochastic Control Theory, {\it Phys.Rev.D},
{\bf 27},1774-1786 (1983)


\bibitem{K} Kantorovich L. On the translocation of masses. C.R. (Doklady) Acad. Sci. URSS (N.S.), 37:199-201, (1942).

\bibitem{L} Loffredo M.I. Eulerian Variational Principle in Stochastic Mechanics, private communication, (1986), Rapporto Matematico 226, Universita' di Siena (1990)


\bibitem{M} Monge G. Memoire sur la theorie des deblais et de remblais, Memoires de l'Academie des Sciences (1781)



\bibitem{Nelson1}Nelson E. Dynamical  Theories  of  Brownian  Motion(Princeton,  NJ:  Princeton  University      Press)(1967)
\bibitem{Nelson*}Nelson E. in Seminaire de Probabilites,Vol. XIX of lecture notes in Mathematics, edited by J.Azema and M Yor, Springer, New York,(1984)

 \bibitem{Nelson2} Nelson E. {\it  Quantum Fluctuations} (Princeton University Press)(1985)
 
\bibitem{Pavon} Pavon M. Hamilton's Principle in stochastic mechanics, J. Math. Phys. 36, 6774 (1995)
\bibitem{Yang} Jianhao M. Yang Variational Principle for  Stochastic Mechanics based on Information Measures J. Math. Phys. 62, 102104 (2021)
\bibitem{Yasue} Yasue K , Stochastic Calculus of Variations, J. of Functional
Analysis 41, 327-340 (1981)

\bibitem{V} Villani C. {\it Topics in Optimal Transportation}, American Mathematical Soc. ISBN     978-0-8218-3312-4,(2003)


\end {thebibliography}
\end {document}